\documentclass[11pt,a4paper]{article}
\usepackage[utf8]{inputenc}
\usepackage{amsmath}
\usepackage{amsfonts}
\usepackage{amssymb}
\usepackage{amsthm}
\usepackage{lmodern}
\usepackage[margin=3cm]{geometry}
\usepackage{authblk}
\usepackage[inline]{enumitem}
\usepackage{multicol}
\usepackage{stmaryrd}
\usepackage{graphicx}
\usepackage{latexsym}
\usepackage{bm}

\newtheorem{theorem}{Theorem}[section]
\newtheorem{lemma}[theorem]{Lemma}
\newtheorem{proposition}[theorem]{Proposition}

\theoremstyle{definition}
\newtheorem{definition}[theorem]{Definition}

\newtheorem{example}[theorem]{Example}

\makeatletter  
\def\th@plain{%
  \thm@notefont{}
  \itshape 
}
\def\th@definition{%
  \thm@notefont{}
  \normalfont 
}
\makeatother

\newcommand{\tot}{\leftrightarrow}

\newcommand{\pa}{\alpha}
\newcommand{\pb}{\beta}

\newcommand{\apa}{\mathrm{a}}
\newcommand{\apb}{\mathrm{b}}
\newcommand{\fp}{\varphi}
\newcommand{\fq}{\psi}

\newcommand{\afp}{\mathrm{p}}
\newcommand{\afq}{\mathrm{q}}

\title{Finitely-valued Propositional Dynamic Logic}

\author{Igor Sedl\'ar\footnote{E-mail: sedlar@cs.cas.cz. This work was
supported by the Czech Science Foundation grant GJ18-19162Y for the project \emph{Non-Classical Logical Models of Information Dynamics}. Some results presented here were obtained while the author was visiting Manuel A.\ Martins and Alexandre Madeira at the University of Aveiro; our discussions helped to initiate this work and shaped the paper. The author is also grateful to Petr Cintula and Libor B\v{e}hounek for discussions on their earlier work on many-valued dynamic logic and the audience of the Seminar of Applied Mathematical Logic at the Institute of Computer Science CAS. The valuable comments of the anonymous reviewers are gratefully acknowledged.

\medskip

\emph{Note:} This is a preprint of the article I. Sedl\'ar: Finitely-valued propositional dynamic logics. In: Nicola Olivetti, Rineke Verbrugge, Sara Negri, Gabriel Sandu (Eds.), Advances in Modal Logic, Volume 13, pp. 561-579. College Publications, 2020.
}}

\affil{The Czech Academy of Sciences, Institute of Computer Science \authorcr 
Pod Vod\'arenskou v\v{e}\v{z}\'i 271/2\authorcr Prague, The Czech Republic}

\date{}

\begin{document}

\maketitle

  \begin{abstract}
We study a many-valued generalization of Propositional Dynamic Logic where formulas in states and accessibility relations between states of a Kripke model are evaluated in a finite FL-algebra. One natural interpretation of this framework is related to reasoning about costs of performing structured actions. We prove that PDL over any finite FL-algebra is decidable. We also establish a general completeness result for a class of PDLs based on commutative integral FL-algebras with canonical constants.

\smallskip

  \end{abstract}

\section{Introduction}
Propositional dynamic logic, PDL, is a well-known modal logic formalizing reasoning about structured actions, e.g.\ computer programs or actions performed by physical agents, and their correctness properties \cite{Fischer1979,Harel2000}.  PDL is subject to two limiting design features. First, being based on classical propositional logic, it formalizes actions that modify values of Boolean variables. A more general setting, one where variables take values from an arbitrary set (integers, characters, trees etc.), is offered by variants of first-order Dynamic Logic, DL \cite{Harel1979,Harel2000}; these variants, however, are mostly undecidable. Second, PDL can express the fact that one action is guaranteed to attain a certain goal while another action is not, but it is not able to express that one action is a more efficient way of attaining the goal than another action. In other words, accessibility between states mediated by actions is modelled as a crisp rather than a graded relation; the former approach is a convenient idealization, but the latter one is more realistic and often also practically required. 

Both of these limitations of ``classical'' PDL are avoided in a \emph{many-valued} setting. In such a setting, values of formulas in states of a Kripke model are taken from an algebra that is typically distinct from the two-element Boolean algebra used in classical PDL. In a many-valued setting, accessibility between states can also be evaluated in such an algebra, naturally leading to a representation of ``costs'' or other ``weights'' associated with performing actions under specific circumstances. 

Research into many-valued modal logics dates back to the 1960s, see the pioneering \cite{Segerberg1967} and the later \cite{Ostermann1988}. Fitting \cite{Fitting1991,Fitting1992} was the first to study modal logics where both formulas in states and accessibility relations between states in the Kripke model take values from a non-Boolean algebra. Fitting considers finite Heyting algebras; generalizations studied for example in \cite{Bou2011,Caicedo2015,Caicedo2010,Hansoul2013,Vidal2017} focus on various kinds of finite or infinite \emph{residuated lattices} \cite{Galatos2007}. Residuated lattices are algebraic structures related to \emph{substructural logics}, with many important special cases such as Boolean and Heyting algebras, relation algebras, lattice-ordered groups, powersets of monoids, various algebras on the $[0,1]$-interval and so on. 

Investigations of PDL based on residuated lattices are relatively scarce. The work in \cite{Behounek2008,Hughes2006,Liau1999} focuses on expressivity of  PDL with many-valued accessibility, but technical results such as decidability or completeness are not provided. Teheux \cite{Teheux2014} establishes decidability and completeness of PDLs based on finite {\L}ukasiewicz chains and the present author \cite{Sedlar2016} establishes decidability and completeness of PDL extending the paraconsistent modal logic of \cite{Odintsov2010}; both papers, however, deal with crisp acessibility relations. As an attempt to sytematize the work in many-valued PDL, Madeira et al.\ \cite{Madeira2016,Madeira2015} put forward a general method of producing many-valued versions of PDL, based on the matrix representation of Kleene algebras; their method, however, applies only if models are defined to be finite. 

In this paper we add to this literature by studying PDLs based on finite Full Lambek algebras, that is, residuated lattices with a distinguished, though arbitrary, $0$ element. We assume that both evaluations of formulas in states and accessibility between states are many-valued. Our main technical results are general completeness and decidability proofs for logics in the family. To the best of our knowledge, our results are the first decidability and completeness results concerning non-crisp many-valued PDL. To be more specific, we work with versions of test-free PDL based on finite Full Lambek algebras with canonical constants; we prove that any PDL based on a finite FL-algebra with canonical constants is decidable; we also establish a completeness result for PDLs based on finite commutative integral FL-algebras with canonical constants.

The paper is structured as follows. Section \ref{sec--PRE} introduces the general framework of PDL based on finite FL-algebras. We note that, for technical reasons discussed in \S\ref{sec--ISS}, our version of PDL uses the transitive closure operator, or Kleene plus, as primitive instead of the more standard reflexive transitive closure operator, the Kleene star. An informal interpretation of the framework is discussed in \S\ref{sec--MOT}. Section \ref{sec--DEC} establishes our decidability result using a generalization of the smallest filtration technique. Section \ref{sec--COM} establishes the completeness result for PDLs based on finite integral commutative FL-algebras with canonical constants. Our work there builds on the results of \cite{Bou2011}, but the canonical model construction used in our proof is novel to this paper (it is a suitable generalization of the greatest filtration construction, though the model itself is infinite). 

\section{Preliminaries}\label{sec--PRE}

In this section we briefly recall two-valued PDL (\S\ref{sec--PRE1}), and we define FL-algebras and many-valued models for the language of PDL based on them (\S\ref{sec--PRE2}). We point out some basic facts that we will use later on.

\subsection{Two-valued PDL}\label{sec--PRE1}

We begin by recalling some well-known facts about two-valued test-free PDL; see \cite{Harel2000}. Fix $Ac = \{ \apa_{i} \mid i \in \omega \}$, a countable set of atomic action expressions. The set of \emph{standard action expressions}, $\mathit{STA}$, is the closure of $Ac$ under applying binary operators $;$ (``composition''), $\cup$ (``choice'') and unary $*$ (``Kleene star''). That is, $\mathit{STA}$ are regular expressions over $Ac$ without the empty expression. For example, $(\apa_0 ; \apa_1)^{*} \cup \apa_0 $ is in $ \mathit{STA}$. Let $Pr = \{ \afp_i \mid i \in \omega \}$ be a countable set of propositional variables. Take $\bm{2}$, the two-element Boolean algebra on the set $\{ 0, 1 \}$ with meet $\sqcap$, join $\sqcup$ and complement $\mathord{-}$; the binary operation $\Rightarrow$ is defined as usual: $a \Rightarrow b := \mathord{-} a \sqcup b$. Formulas of the \emph{standard language for $\bm{2}$}, $Fm(\mathcal{L}^{\mathit{STA}}_{\bm{2}})$, are defined by
\[
\fp \: := \:  \afp \mid \bar{c} \mid \fp \land \fp \mid \fp \lor \fp \mid \fp \to \fp \mid [\pa]\fp
\] 
where $\afp \in Pr$, $c \in \bm{2}$ and $\pa \in \mathit{STA} $. For example, $\afp_0 \to [\apa_0 ; (\apa_1)^{*}] (\afp_1 \to \bar{0})$ is a formula of $\mathcal{L}^{\mathit{STA}}_{\bm{2}}$.

A \emph{$\bm{2}$-valued frame for $\mathit{STA}$} is $\mathfrak{F} = (S, \{ R_{\pa} \}_{\pa \in \mathit{STA}})$ where $S$ is a non-empty set and, for each $\pa \in \mathit{STA}$, $R_{\pa}$ is a function from $S \times S$ to $\bm{2}$. We denote $R(\pa) := \{ (s,t) \mid R_{\pa}(s,t) = 1 \}$; and the functions in $\{ R_{\pa} \}_{\pa \in \mathit{STA}}$ are required to satisfy the following:
\begin{enumerate*}
\item $R(\pa \cup \pb) = R(\pa) \cup R(\pb)$;
\item $R(\pa ; \pb) = R(\pa) \circ R(\pb)$, the composition of $R(\pa)$ and $R(\pb)$;
\item $R(\pa^{*}) = R(\pa)^{*}$, the reflexive transitive closure of $R(\pa)$.
\end{enumerate*}

Let $\mathfrak{F} = (S, \{ R_{\pa} \}_{\pa \in \mathit{STA}})$ be a $\bm{2}$-valued frame. A \emph{$\bm{2}$-valued model based on $\mathfrak{F}$} is $\mathfrak{M} = (S, \{ R_{\pa} \}_{\pa \in \mathit{STA}}, V)$ where $V : Fm(\mathcal{L}^{\mathit{STA}}_{\bm{2}}) \times S \to \bm{2}$ such that
\begin{itemize}
\item $V(\bar{c}, s) = c$;
\item $V(\fp \land \fq, s) = V(\fp, s) \sqcap V(\fq, s)$, $V(\fp \lor \fq, s) = V(\fp, s) \sqcup V(\fq, s)$, and $V(\fp \to \fq, s) = V(\fp, s) \Rightarrow V(\fq, s)$;
\item $V([\pa]\fp, s) = \bigsqcap_{t \in S} \big (R_{\pa} (s, t) \Rightarrow V(\fp, t) \big )$.
\end{itemize}
\noindent Note that $V([\pa]\fp, s) = \bigsqcap_{R_{\pa}(s, t) = 1} V(\fp,t)$. A formula $\fp$ is \emph{valid in $\mathfrak{M}$} iff $V(\fp,s) = 1$ for all $s$; validity in frames and classes of frames is defined as expected.

This is the standard presentation of test-free PDL, phrased in a way that invites generalizations obtained by replacing $\bm{2}$ by another algebra. We will study some such generalizations in this paper but, as we discuss in more detail below, the story is somewhat more complicated. For reasons discussed in \S\ref{sec--ISS}, our generalizations will use a different primitive iteration operator instead of the Kleene star. The operator we will use, however, is conveniently related to the Kleene star.

The set of \emph{action expressions over $Ac$}, $\mathit{ACT}$, is the closure of $Ac$ under composition, choice and the unary operator $+$ (``Kleene plus''). Formulas of the \emph{language $\mathcal{L}_{\bm{2}}$} are defined as expected (we omit reference to $\mathit{ACT}$), with $\pa \in \mathit{ACT}$; for example, $\afp_0 \to [\apa_0 ; (\apa_1)^{+}] (\afp_1 \to \bar{0})$ is a formula of $\mathcal{L}_{\bm{2}}$. The definition of \emph{$\bm{2}$-valued frames for $\mathit{ACT}$} is the same as the definition of $\bm{2}$-valued frames for $\mathit{STA}$, with an obvious exception, namely, the requirement that $R(\pa^{+})$ be the \emph{transitive closure} of $R(\pa)$, i.e.\ $R(\pa^{+}) = \bigcup_{n > 0} R^{n}(\pa)$, where $R^{1}(\pa) = R(\pa)$ and $R^{n+1}(\pa) = R^{n}(\pa) \circ R(\pa)$. Compare this with the \emph{reflexive} transitive closure $R(\pa)^{*} = \bigcup_{n \geq 0} R^{n}(\pa)$, where $R^{0}(\pa) = \{ (s,s) \mid s \in S \}$. Models based on frames for $\mathit{ACT}$ are defined as before. 

\begin{proposition}\label{prop--PRE-Avoiding Kleene star}
Let For each $\pa \in \mathit{ACT}$ and $\fp \in Fm(\mathcal{L}_{\bm{2}})$, \[V(\fp\land [\pa^{+}]\fp, s) = 1 \quad \text{iff} \quad
\forall t ( (s,t) \in R(\alpha)^{*} \implies V(\fp, t) = 1)\, . \]
\end{proposition} 

Proposition \ref{prop--PRE-Avoiding Kleene star} implies that $\fp \land [\pa^{+}]\fp$ ``simulates'' $[\pa^{*}]\fp$ in $\mathcal{L}_{\bm{2}}$. This provides a justification for our using languages based on $\mathit{ACT}$ rather than on $\mathit{STA}$ in what follows.  However, we admit that this choice is related to the technical issues discussed in \S\ref{sec--ISS}.

\subsection{FL-algebras and finitely-valued PDL}\label{sec--PRE2}
In this section we generalize two-valued PDL by replacing the two-element Boolean algebra $\bm{2}$ by a more general structure, namely, a finite FL-algebra. FL-algebras provide semantics for a wide class of \emph{substructural logics} \cite{Galatos2007}.

\begin{definition}
An \emph{FL-algebra} (``full Lambek algebra'', \cite{Galatos2007}) is a set $X$ with binary operations $\sqcap, \sqcup, \backslash, \cdot, \slash$ and two distinguished elements $1, 0$ such that
\begin{itemize}
\item $(X, \sqcap, \sqcup)$ is a lattice (let $a \sqsubseteq b$ iff $a \sqcup b = b$);
\item $(X, \cdot, 1)$ is a monoid;
\item $(\backslash, \cdot, \slash)$ are residuated over $(X, \sqsubseteq)$, i.e.\
	\[
	a \cdot b \sqsubseteq c \quad\text{ iff }\quad 
	b \sqsubseteq a \backslash c \quad\text{ iff }\quad 
	a \sqsubseteq c \slash b \, ;
	\]
\item $0$ is an arbitrary element of $X$.
\end{itemize}
\emph{Residuated lattices} are $0$-free reducts of FL-algebras. Each finite FL-algebra $\bm{X}$ contains a least element $\bot^{\bm{X}}$ (for all $a \in X$, $\bot^{\bm{X}} \sqsubseteq a$) and a greatest element $\top^{\bm{X}}$ (for all $a \in X$, $a \sqsubseteq \top^{\bm{X}}$). 
\end{definition}

We usually write $ab$ instead of $a \cdot b$ and $a \Rightarrow b$ instead of $b \slash a$. Two varieties of FL-algebras will be important in this paper:
\begin{itemize}
\item \emph{commutative} FL-algebras satisfy $ab = ba$ for all $a,b \in X$;
\item \emph{integral} FL-algebras satisfy $a \sqsubseteq 1$ for all $a \in X$.
\end{itemize}
\noindent Note that in commutative FL-algebras $a \backslash b = b \slash a$.

\begin{example}
The two-element Boolean algebra $\bm{2}$ is a commutative integral FL-algebra, where $\cdot$ is $\sqcap$ and $\backslash$ (identical to $\slash$) is $\Rightarrow$. 
\end{example}
\begin{example}\label{exam:Luk}
Let $N > 0$ and define $\bm{N} = (N, max, min, +_N, \to_N )$ where 
\[ 
a +_N b = min (a + b, N - 1) \quad \text{and} \quad
a \to_N b = max (b - a, 0) \, .
\]
$\bm{N}$ is a finite commutative integral FL-algebra, with $0$ as the monoid identity with respect to $+_{N}$ and the greatest element under the $\geq$-ordering induced by taking $min$ as join. We note that $\bm{N}$ is isomorphic to the $N$-element {\L}ukasiewicz lattice $\textit{\bm{{\L}}}_{N}$ over $\{ \frac{k}{N-1} \mid k \in N \}$.
\end{example}
\begin{example}
As an example of a non-commutative, non-integral infinite FL-algebra, take the power set of the free monoid over some set $\Sigma$, i.e.\ the set of languages over $\Sigma$, with intersection as meet, union as join, $L \cdot L' := \{ xx' \mid x \in L \And x' \in L' \}$, $\{ \varepsilon \}$ as the monoid identity ($\varepsilon$ is the empty word) and $L \backslash L' := \{ x \in \Sigma \mid L \cdot \{ x \} \subseteq L' \}$, $L' \slash L := \{ x \in \Sigma \mid \{ x \} \cdot L \subseteq L' \}$.
\end{example}

The following lemma summarizes some of the properties of FL-algebras we will rely on in this paper (we will often say that something holds ``by the properties FL-algebras'' in our proofs).

\begin{lemma}\label{lem--APP-Properties of L}
Let $\bm{X}$ be an arbitrary FL-algebra. Then
\begin{enumerate*}
\item $a \sqsubseteq b$ iff $1 \sqsubseteq a \Rightarrow b$;
\item If $a \sqsubseteq b$ and $c \sqsubseteq d$, then $b \Rightarrow c \sqsubseteq a \Rightarrow d$, $b \backslash c \sqsubseteq a \backslash d$ and $ac \sqsubseteq bd$;
\item $(a \sqcup b)c = ac \sqcup ab$ and $c(a \sqcup b) = ca \sqcup cb$;
\item $a \Rightarrow (b \sqcap c) = (a \Rightarrow b) \sqcap (a \Rightarrow c)$;
\item $a \sqcup b \Rightarrow c = (a \Rightarrow c) \sqcap (b \Rightarrow c)$;
\item $a \Rightarrow (b \Rightarrow c) = ab \Rightarrow c$;
\item $(a \Rightarrow b)(b \Rightarrow c) \sqsubseteq a \Rightarrow c$;
\item $(1 \Rightarrow a) = a$
\end{enumerate*}
\end{lemma}

If $S$ is a non-empty set, then $\Pi(S)$ is the set of all finite sequences of elements of $S$; that is, $\pi \in \Pi(S)$ iff $\pi$ is a function from some $n \in \omega$, called the length of $\pi$, to $S$. The unique sequence of length $0$ is $\emptyset$. If $\pi$ is a sequence of length $n$ and $s \in S$, then $\pi^{\frown}s$ is the unique sequence of length $n+1$ such that $(\pi^{\frown}s)(k) = \pi(k)$ for all $k < n$ and $(\pi^{\frown}s)(n) = s$. Note that each sequence $\pi$ of length $n > 0$ can be expressed as $( \ldots (\emptyset^{\frown}\pi(0))^{\frown} \ldots)^{\frown}\pi(n-1)$. 

\begin{definition}
Let $\bm{X}$ be a finite FL-algebra and $S$ a non-empty set. A \emph{binary $\bm{X}$-valued relation on $S$} is any function from $S \times S$ to $\bm{X}$. Let $R, Q$ be binary $\bm{X}$-valued relations on a set $S$; then
\begin{itemize}
\item the \emph{union} of $R$ and $Q$ is the function $R \cup Q$ defined by $(R \cup Q)(s,t) : = R(s,t) \sqcup Q(s,t)$;
\item the \emph{composition} of $R$ and $Q$ is the function $R \circ Q$ defined by $(R \circ Q)(s,t) = \bigsqcup_{x \in S} \big ( R(s,x) \cdot Q(x,t) \big)$;
\item the \emph{transitive closure} of $R$ is the function $R^{+}$ defined by $R^{+}(s,t) = \bigsqcup_{\pi \in \Pi(S)} R s \pi t$ where $R s \pi t$ is defined as follows:
	\begin{itemize}
	 \item $R s \emptyset t = R (s, t)$ and
	 \item $R s (\pi^{\frown}u) t = R s \pi u \cdot R(u,t)$.
	 \end{itemize} 
\end{itemize}
We say that $Q$ \emph{extends} $R$, notation $R \sqsubseteq Q$, iff $R(s,t) \sqsubseteq Q(s,t)$ for all $s,t \in S$; $R$ is the \emph{smallest} relation in a set $\{ R_i \}_{i \in I}$ if $R = R_i$ for some $i \in I$ and each $R_i$ extends $R$. $R$ is \emph{transitive} if $R(s,t) \cdot R(t,u) \sqsubseteq R(s,u)$ for all $s,t,u \in S$; and $R$ is \emph{reflexive} if $1 \sqsubseteq R(s,s)$ for all $s \in S$.
\end{definition}
\noindent Note that we need to assume that all the required joins exist in $\bm{X}$; hence the restriction to finite FL-algebras (however, a restriction to complete $\bm{X}$ is sufficient, as is the assumption that $R, Q$ are ``$\bm{X}$-safe'' \cite[ch.\ 5]{Hajek1998}).

\begin{proposition}\label{prop--PRE-closure}
Let $\bm{X}$ be a finite FL-algebra and $R$ a binary $\bm{X}$-valued relation on a set $S$. Then $R^{+}$ is the smallest transitive relation extending $R$. For any $R$,  define $R^{*}$ as follows:
\[
R^{*}(s,t) \: = \: \begin{cases}
1 & \text{if } s = t \\
R^{+}(s,t)  & \text{otherwise.}
\end{cases}
\] Then $R^{*}$ is the smallest reflexive transitive relation extending $R$.
\end{proposition}
\begin{proof}
It is clear that $R^{+}$ is a transitive relation extending $R$. Now assume that so is $Q$. The conclusion that $R^{+} \sqsubseteq Q$ follows from two facts that are easily established by induction on the length of $\pi$: (a) For all $s,t \in S$ and $\pi \in \Pi(S)$, $R s \pi t \sqsubseteq Q s \pi t$ (the assumption that $R \sqsubseteq Q$ is used here); (b) For all $s,t \in S$ and $\pi \in \Pi(S)$, $Q s \pi t \sqsubseteq Q(s,t)$ (the assumption that $Q$ is transitive is used). Since $\bm{X}$ is finite, the two claims imply that, for any given $s$ and $t$, $\bigsqcup_{\pi} Rs\pi t \sqsubseteq \bigsqcup_{\pi} Q s \pi t \sqsubseteq Q(s,t)$.

It is clear that $R^{*}$ is a reflexive transitive relation extending $R$. If so is $Q$, then we reason for any given $s$ and $t$ by cases as follows. If $s = t$, then $R^{*}(s,t) \sqsubseteq Q(s,t)$ is equivalent to $1 \sqsubseteq Q(s,s)$, which holds by reflexivity of $Q$. If $s \neq t$, then $R^{*}(s,t) \sqsubseteq Q(s,t)$ is equivalent to $R^{+}(s,t) \sqsubseteq Q(s,t)$, which follows from the assumption that $Q$ is a transitive relation extending $R$. Hence, $R^{*}(s,t) \sqsubseteq Q(s,t)$ for any $s$ and $t$.
\end{proof}

\begin{lemma}\label{lem--PRE-Identity}
Let $\bm{X}$ be a finite FL-algebra and $S$ a set; the $\bm{X}$-valued identity relation on $S$ is defined as follows:
\[
Id_{\bm{X}}(s,t) := \begin{cases}
1 & \text{if } s = t \\
\bot^{\bm{X}} & \text{otherwise.}
\end{cases}
\] If $\bm{X}$ is integral, then $R^{*} = Id_{\bm{X}} \cup R^{+}$ for any binary $\bm{X}$-valued relation on $S$.
\end{lemma}
\begin{proof}
We omit the proof; we just note that if $s = t$, then $R^{*}(s,t) = Id_{\bm{X}}(s,t) \sqcup R^{+}(s,t)$ is equivalent to $R^{+}(s,t) \sqsubseteq 1$, which is guaranteed to hold only if $\bm{X}$ is integral.
\end{proof}

\begin{definition}\label{def--frame}
Let $\bm{X}$ be a finite FL-algebra. An \emph{$\bm{X}$-valued frame for $\mathit{ACT}$} is a pair $\mathfrak{F} = (S, \{ R_{\pa} \}_{\pa \in \mathit{ACT}})$ where $S$ is a non-empty set and, for all $\pa \in \mathit{ACT}$, $R_{\pa}$ is an $\bm{X}$-valued binary relation on $S$ such that
\begin{enumerate*}
\item $R_{\pa \cup \pb} = R_{\pa} \cup R_{\pb}$;
\item $R_{\pa ; \pb} = R_{\pa} \circ R_{\pb}$; and
\item $R_{\pa^{+}} = R_{\pa}^{+}$.
\end{enumerate*}
\end{definition}
\noindent $\bm{X}$-valued frames will also be referred to as $\bm{X}$-frames or simply frames if $\bm{X}$ is clear from the context or immaterial. We will sometimes write $R_{\pa}st$ instead of $R_{\pa}(s,t)$.

\begin{definition}\label{def--formulas}
Formulas of the \emph{language $\mathcal{L}_{\bm{X}}$} are defined as follows:
\[
\fp := \afp \mid \bar{c} \mid \fp \land \fp \mid \fp \lor \fp \mid \fp \backslash \fp \mid \fp \cdot \fp \mid \fp \slash \fp \mid [\pa] \fp \, ,
\]
where $\afp \in Pr$, $c \in \bm{X}$ and $\pa \in \mathit{ACT}$. We use $\bot, \top$ instead of $\overline{\bot^{\bm{X}}}$ and $\overline{\top^{\bm{X}}}$, respectively. We often write $\fp\fq$ instead of $\fp \cdot \fq$, $\fp \to \fq$ instead of $\fq \slash \fp$, $m$ instead of $\apa_m$, and $\pa\pb$ instead of $\pa; \pb$. We define $\fp \tot \fq := (\fp \to \fq) \land (\fq \to \fp)$, $\neg \fp := \fp \to \bot$ and $\langle \pa \rangle \fp := \neg [\pa] \neg \fp$.
\end{definition}
Note that we use the same symbol $\otimes \in \{ \backslash, \cdot, \slash \}$ for the implication and fusion connectives of the language and for the residuated operations on FL-algebras. We will denote the operations on a given $\bm{X}$ as $\otimes^{\bm{X}}$ in contexts where it is convenient for the reader to distinguish the connectives of the language from the operations on the algebra. (However, $\Rightarrow$ denotes the operation $\slash^{\bm{X}}$ and $\to$ denotes the connective $\slash$ throughout.)

\begin{definition}\label{def--model}
A \emph{model} based on an $\bm{X}$-frame $(S, \{ R_{\pa} \}_{\pa \in \mathit{ACT}})$ is $\mathfrak{M} = (S, \{ R_{\pa} \}_{\pa \in \mathit{ACT}}, V)$, where $V $ is a function from $Fm(\mathcal{L}_{\bm{X}}) \times S$ to $\bm{X}$ such that
\begin{itemize}
\item $V(\bar{c}, s) = c$;
\item $V(\fp \land \fq, s) = V(\fp, s) \sqcap V(\fq, s)$ and $V(\fp \lor \fq, s) = V(\fp, s) \sqcup V(\fq, s)$;
\item $V(\fp \otimes \fq, s) = V(\fp, s) \otimes^{\bm{X}} V(\fq, s)$ for $\otimes \in \{ \backslash, \cdot, \slash \}$;
\item $V([\pa]\fp, s) = \bigsqcap_{t \in S} \big ( R_{\pa} st \Rightarrow V(\fp, t) \big )$.
\end{itemize}
\noindent A formula $\fp$ is \emph{valid in $\mathfrak{M}$} iff $1 \sqsubseteq V(\fp, s)$ for all $s$ in $\mathfrak{M}$. Validity in frames and classes of frames is defined as expected. The \emph{theory} of a frame is the set of formulas valid in the frame; the theory of a class of frames is the set of formulas valid in each frame in the class. $Th(\bm{X})$ is the theory of the class of all $\bm{X}$-frames.
\end{definition}

The following addendum to Proposition \ref{prop--PRE-Avoiding Kleene star} suggests that integral FL-algebras are particularly suitable for us.

\begin{proposition}\label{prop--PRE-Reflexive transitive closure}
Take an arbitrary $\bm{X}$-frame for a finite integral $\bm{X}$. Then $V(\fp \land [\pa^{+}]\fp, s) = \bigsqcap_{t \in S} (R^{*}_{\pa} st \Rightarrow V(\fp,t))$.
\end{proposition}
\begin{proof}
The $\sqsubseteq$-inequality is straightforward and the $\sqsupseteq$-inequality follows from Lemma \ref{lem--PRE-Identity}. 
\end{proof}

It is clear that two-valued PDL is a special case of the present framework for $\bm{X} = \bm{2}$. 

\begin{lemma}\label{lem--PRE-PDL validities hold}
The following are valid in each $\bm{X}$-frame:
\begin{multicols}{2}
\begin{enumerate}[label={(\alph*)}]
\item $[\pa] (\fp \land \fq) \tot ([\pa]\fp \land [\pa]\fq)$
\item $[\pa \cup \pb]\fp \tot ([\pa]\fp \land [\pb]\fp)$
\item $[\pa\pb] \fp \tot [\pa][\pb]\fp$
\item $[\pa^{+}]\fp \tot [\pa](\fp \land [\pa^{+}]\fp)$
\end{enumerate}
\end{multicols}
\end{lemma}
\begin{proof}
To prove that $\fp \tot \fq$ is valid if suffices to show that $V(\fp, s) = V(\fq,s)$ for all $s$ in all models. (a) The proof relies on the fact that $a \Rightarrow (b \sqcap c) = (a \Rightarrow b) \sqcap (a \Rightarrow c)$ in all FL-algebras. (b) The proof relies on the fact that $(a \sqcup b) \Rightarrow c = (a \Rightarrow c) \sqcap (b \Rightarrow c)$ in all FL-algebras. (c) The proof relies on the fact that $a \Rightarrow (b \Rightarrow c) = ab \Rightarrow c$ in all FL-algebras. (Note that composition of relations needs to be defined using monoid multiplication $\cdot$, not lattice meet.) (d) The proof relies on the fact that $R_{\pa} st \sqsubseteq R_{\pa^{+}} st$, it also uses simple composition of paths.
\end{proof}

We will discuss an informal interpretation of a special case of the many-valued framework in the next section. Speaking generally, however, we may adapt the slogan characterizing modal logic as providing languages for talking about relational structures \cite[p.\ viii]{Blackburn2001} and say that many-valued modal logics provide \emph{simple yet expressive languages for talking about many-valued relational structures}. Examples of many-valued relational structures include weighted structures such as weighted graphs etc. Choosing an FL-algebra as the algebra of weights brings the framework closer to substructural logics that include well-known formalisms for reasoning about resources (variants of linear logic) or graded properties and relations (fuzzy logics). Many-valued PDL adds to this the capacity to articulate reasoning about \emph{structured} many-valued relations using the PDL relational operations of choice, composition and iteration. An intriguing connection here is the relation of finitely-valued PDL to weighted automata over finite semirings \cite{Droste2009}, but a more thorough investigation of this connection is left for another occasion.

\section{Motivation}\label{sec--MOT}

This section discusses the informal interpretation of finitely-valued PDL. We give two general interpretations of the framework first and then we zoom in to PDLs over a specific class of FL-algebras. Our overview is cursory; the present paper is focused more on basic technical results than on informal interpretations and applications. A more thorough exploration of the latter is left for another occasion. We only note here that we consider many-valued PDL to be sufficiently mathematically interesting to be studied independently of informal interpretations and applications.

We have mentioned before the slogan that modal logics provide simple yet expressive languages for talking about relational structures \cite[p.\ viii]{Blackburn2001}; by the same token, many-valued modal logics can be seen as providing means of talking about ``weighted'' relational structures. Two-valued PDL has been applied to at least two kinds of relational structures which have very natural weighted generalizations. We discuss these in turn.

First, take the interpretation of modal logic that relates it to \emph{description logics} \cite{Baader2007}. Simply put, formulas of a modal language can be seen as expressing ``concepts'', i.e.\ properties of objects, and indices of modal operators as expressing various ``roles'', i.e.\ relations between objects. On this reading, ``states'' in a Kripke model represent arbitrary objects and ``accessibility relations'' between them represent relations between these objects. Structured modal indices that come with PDL (i.e.\ ``action expressions'' as we call them) can be seen as expressing structured relations between objects; union, composition and transitive closure have been found particularly suitable for expressing various important concepts and roles \cite{Baader1991}. Many-valued description logics (see \cite{Straccia2006} for instance) are a generalization of description logics designed for management of \emph{uncertain and imprecise information}. These logics can express the fact that an object is subsumed under a given concept (e.g.\ ``tall'' if the reader will forgive the platitudinous example) only to some degree or that only imprecise information about a relation holding between two objects is available. Finitely-valued PDL as presented here can be seen as a family of many-valued description logics with transitive closure of roles.

Second, the original motivation of PDL was reasoning about the behaviour of computer programs \cite{Fischer1979}. From a more general perspective, PDL can be seen as a logic formalising reasoning about types of \emph{structured actions}, represented by ``action expressions''. On this reading, a Kripke frame consists of states and transitions between states labelled by types of action; for instance $R_{\pa}st$ means that action of type $\pa$ can be used to get from state $s$ to state $t$. States can be thought of as physical locations, states of a complex system such as a database or states of a computer during the run of a program; but states can also be thought of as ``states of the world'' that can be modified by actions of intelligent agents. PDL can be used to formalize reasoning about properties of actions that modify these kinds of states. One important example is correctness, related to the question if a specific kind of action is guaranteed to lead to a specific outcome when performed under specific circumstances. (This more general perspective makes PDL relevant to automated planning, for example.) Many-valued Kripke models can be seen as transition systems where transitions carry \emph{weights}; these can be costs or resources needed to perform a transition using the given action type. B\v{e}hounek \cite{Behounek2008} suggested a many-valued version of PDL for reasoning about costs of program runs that is close to our framework, but he did not establish completeness or decidability results.

Let us now discuss a special case of the finitely-valued PDL framework giving rise to a natural class of weighted relational structures; we show that formulas of the PDL language are able to express interesting features of these structures. Let $\bm{N}$ be the FL-algebra of Example \ref{exam:Luk}, that is, $\bm{N} = (N, max, min, +_N, \to_N )$ where 
\[ 
a +_N b = min (a + b, N - 1) \quad \text{and} \quad
a \to_N b = max (b - a, 0) \, ,
\]
where $N \in \omega$ is non-empty. The set $N$ is seen as a \emph{weight scale} with $0$ representing zero weight (``for free'') and $N-1$ representing the maximal weight (considered ``infeasible''). The operation $+_{N}$, namely, sum bounded by the maximal weight, represents \emph{weight addition}. $N$ is given a (distributive) lattice structure by including $max$ as meet and $min$ as join; the associated lattice order $\sqsubseteq$ is defined as usual, $a \sqsubseteq b$ iff $min(a,b) = b$. Hence, $a \sqsubseteq b$ (i.e.\ $b \leq a$) means that weight $b$ is at most as big as weight $a$. The choice of $max$ as meet and $min$ as join---not the other way around---may seem unintuitive at first, but it yields the result that $a \sqsubseteq 0$ for all $a \in N$. It is important to note in this respect that $0$ is the identity element with respect to $+_{N}$. (Hence, choosing the natural ordering on $N$ as our lattice ordering would mean that each element of the lattice would be above the monoid identity, which is problematic given our definition of validity.) It is clear that $ a +_{N} b = b +_{N} a$. The residual $\to_{N}$ of $+_{N}$ is truncated subtraction or monus; the crucial feature of $\to_{N}$ is that $a \to_{N} b = 0$ iff $a \sqsubseteq b$ (iff $b \sqsubseteq a$). We note that $\bm{N}$ is isomorphic to the $N$-element {\L}ukasiewicz lattice $\textit{\bm{{\L}}}_{N}$ over $\{ \frac{k}{N-1} \mid k \in N \}$, but we prefer $\bm{N}$ to $\textit{\bm{{\L}}}_{N}$ as a representation of an $N$-element weight scale.

$\bm{N}$-frames are weighted relational structures that can be informally interpreted in a number of ways. On the ``description reading'', for instance, states $s \in S$ are objects and $R_{\pa}$ represent structured weighted relations between these objects. On the ``transition cost reading'', states can be seen as physical locations or states of a system and $R_\pa st \in N$ is the cost of accessing state $t$ from $s$ by performing action $\pa$ (hence, frames are weighted labelled transition systems). If $R_{\pa} st = N - 1$, then we say that $t$ is not in relation $\pa$ with $s$, or that $t$ cannot be accessed from $s$ by performing $\pa$; if $R_{\pa} st = 0$, then $t$ is ``clearly'' in relation $\pa$ with $s$, or $t$ can be accessed from $s$ by $\pa$ for free. Let us now discuss some properties of weighted relational structures that can be expressed by PDL formulas.

Since $\bm{N}$ is $(N-1)$-involutive, i.e.\ $(a \Rightarrow (N-1)) \Rightarrow (N-1) = a$ for all $a \in \bm{N}$, we have \[ V(\langle \pa\rangle\bar{0}, s) = \bigsqcup_{t \in S} \big( R_{\pa} st +_{N} 0\big) = min \big \{ R_{\pa} st \mid t \in S \big \} \, .\] In other words, $V(\langle \pa\rangle\bar{0}, s)$ is the minimal guaranteed cost of performing $\pa$ at $s$ (on the transition cost reading) or the maximal degree to which $s$ is $\pa$-related to any object (on the description reading). Let us write simply $\pa$ instead of $\langle\pa\rangle\bar{0}$ if the context clears up any possible confusion. Note that $a \Rightarrow b$ is the difference between $b$ and $a$ if $a < b$ and $0$ otherwise. The following features of weighted relation structures can be expressed (we use the transition cost reading and the reader is invited to translate to the description reading):
\begin{itemize}
\item the minimal cost of performing $\pa$ is at most $m$ (this is true in state $s$ if $V(\bar{m} \to \pa, s) = 0$); the ``at least'' direction is expressed dually;
\item performing $\pa$ is at least as costly as performing $\pb$ (this is true in state $s$ if $V(\pa \to \pb, s) = 0$); the ``at most'' direction is expressed dually;
\item the difference between the minimal guaranteed cost of $\pb$ and $\pa$ is at most $m$ (this is true in state $s$ if $V(\bar{m} \to (\pa \to \pb), s) = 0$).
\end{itemize}

On the transition cost reading, atomic formulas in $Pr$ can be seen as representing various items that can be obtained at states for a given cost, with $V(\afp, s)$ representing the cost of item $\afp$ at $s$ (e.g.\ time needed to charge the battery at the charger location). Observe that $V(\langle\pa\rangle \afp, s) = \bigsqcup_{t \in S} (R_{\pa} st +_{N} V(\afp, t))$ is the minimal cost of getting from $s$ to a state $t$ by performing $\pa$ and obtaining $\afp$ at $t$; we may also say that this is the minimal guaranteed cost of obtaining $\afp$ by $\pa$. On the description reading, atomic formulas can be seen as expressing graded, imprecise or vague properties of objects; thus the value of $\langle \pa\rangle \afp$ at $s$ is the ``grade of truth'' of the statement that $s$ is $\pa$-related to an object with property $\afp$. The interesting case obtains where both the relation and the property are graded or vague; think of ``Alice was in contact with a person displaying symptoms of COVID-19''. We write $\fp^{\pa}$ instead of $\langle \pa\rangle \fp$. The following features of weighted relation structures can be expressed (we use the transition cost reading and the reader is again invited to translate to the description reading):
\begin{itemize}
\item the minimal cost of obtaining $\afp$ by $\pa$ is at most $m$ (this is true in state $s$ if $V(\bar{m} \to \afp^{\pa}, s) = 0$); the ``at least'' direction is expressed dually;
\item obtaining $\afp$ by $\pa$ is at least as costly as obtaining $\afq$ by $\pb$ (this is true in state $s$ if $V(\afp^\pa \to \afq^\pb, s) = 0$); the ``at most'' direction is expressed dually;
\item the difference between the minimal guaranteed cost of obtaining $\afq$ by $\pb$ and obtaining $\afp$ by $\pa$ is at most $m$ (this is true in state $s$ if $V(\bar{m} \to (\afp^\pa \to \afq^\pb), s) = 0$).
\end{itemize}

This cursory overview shows that the PDL language provides means to expressing a variety of features of weighted relational structures and so finitely-valued PDL can be used to formalize reasoning about these features. A more thorough exploration of expressivitiy and applications is left for another occasion.

\section{Finite model property and decidability}\label{sec--DEC}

In this section we prove that $Th(\bm{X})$ is decidable for all finite $\bm{X}$. We prove this by showing that each such $Th(\bm{X})$ has the bounded finite model property. The result is established using a many-valued generalization of the smallest filtration construction; see \cite{ConradieEtAl2017}, where the construction is applied to some many-valued modal logics with $\Box$ and $\Diamond$.\footnote{We are grateful to an anonymous reviewer for pointing the reference out.} Even though the decidability result is not surprising, we consider it to be a ``sanity check'' for the many-valued dynamic framework. We note that presence of canonical constants is not necessary for the decidability result (in contrast to the completeness result of \S\ref{sec--COM}).

\begin{definition}
The \emph{closure} of a set of formulas $\Psi$ is the smallest $\Phi \supseteq \Psi$ such that
\begin{itemize}
\item $\Phi$ is closed under subformulas (that is, if $\fp \in \Phi$ and $\fq$ is a subformula of $\fp$, then $\fq \in \Phi$);
\item $[\pa \cup \pb]\fp \in \Phi$ implies $[\pa]\fp \in \Phi$ and $[\pb]\fp \in \Phi$;
\item $[\pa\pb]\fp \in \Phi$ implies $[\pa][\pb]\fp \in \Phi$;
\item $[\pa^{+}]\fp \in \Phi$ implies $[\pa][\pa^{+}]\fp \in \Phi$ and $[\pa]\fp \in \Phi$.
\end{itemize}
$\Phi$ is \emph{closed} iff $\Phi$ is the closure of $\Phi$.
\end{definition}

\begin{definition}
For each set of formulas $\Phi$ and each model $\mathfrak{M}$, we define the binary two-valued equivalence relation $\approx_{\Phi}$ on states of $\mathfrak{M}$ by \[ s \approx_{\Phi} t \: \iff \: ( \forall \fp \in \Phi) \big ( V(\fp, s) = V(\fp,t)\big ) \, .\] The equivalence class of $s$ under $\approx_{\Phi}$ will be denoted as $[s]_{\Phi}$ or just as $[s]$ if $\Phi$ is clear from the context.
\end{definition}

\begin{definition}
Take an $\bm{X}$-valued model $\mathfrak{M}$ and a finite closed set $\Phi$. The \emph{filtration of $\mathfrak{M}$ through $\Phi$} is the $\bm{X}$-valued model $\mathfrak{M}^{\Phi} = (S^{\Phi}, R^{\Phi}, V^{\Phi})$ such that
\begin{itemize}
\item $S^{\Phi} = \{ [s] \mid s \in S \}$;
\item $R^{\Phi}_{\apa_m}([s],[t]) = \bigsqcup \big \{ R_{\apa_m}(u,v) \mid s \approx_{\Phi} u \And t \approx_{\Phi} v \big \}$; $R^{\Phi}_{\pa}$ for $\pa \notin Ac$ is defined as in models;
\item $V^{\Phi}(\afp, [s]) = V(\afp, s)$ for $\afp \in \Phi$; $V^{\Phi}(\afp, [s]) = 0^{\bm{X}}$ for $\afp \notin \Phi$; $V^{\Phi}(\fp, [s])$ for $\fp \notin Pr$ is defined as in models. 
\end{itemize}
\end{definition}

It is clear that if $\Phi$ is the closure of a finite set $\Psi$, then $\Phi$ is finite. If $\Phi$  is finite, then so is $\mathfrak{M}^{\Phi}$; in fact, $|S^{\Phi}| \leq |\bm{X}|^{|\Phi|}$. We usually omit reference to $\Phi$ while discussing accessibility relations on $S^{\Phi}$ and we also write $\approx$ instead of $\approx_{\Phi}$. We will write $R_m$ instead of $R_{\apa_m}$. In the rest of the section, we fix an $\bm{X}$-model $\mathfrak{M}$ and a finite closed set $\Phi$.

\begin{lemma}\label{lem--DEC-filtration lemma}
For all $\pa \in \mathit{ACT}$ and all $x,y \in S$,
\begin{enumerate}[label={(\alph*)}]
\item $R_{\pa}xy \sqsubseteq R_{\pa}[x][y]$;
\item For all $[\pa]\fp \in \Phi$, $V([\pa]\fp, x) \sqsubseteq R_{\pa}[x][y] \Rightarrow V(\fp, y)$.
\end{enumerate}
\end{lemma}
\begin{proof}
Both claims are established by induction on the complexity of $\pa$. The base case of (a) holds by definition and the rest is established easily using the induction hypothesis. In the case of $\pa = \pb^{+}$, we define for each $\pi \in \Pi(S)$ of length $n$ the sequence $[\pi] \in \Pi (S^{\Phi})$ of length $n$ by $[\pi](k) := [\pi(k)]$ for all $k < n$; it is then easy to establish by induction on $n$ that $R_{\pb} x \pi y \sqsubseteq R_{\pb} [x] [\pi] [y]$.)

The base case of (b) is follows from the fact that, for all $x' \in [x]$ and $y' \in [y]$, $\bigsqcap_{z \in S} \big( R_{m} x' z \Rightarrow V(\fp, z) \big) \cdot R_{m} x'y' \sqsubseteq V(\fp, y')$ using the definition of $\approx_{\Phi}$, closure of $\Phi$ under subformulas and properties of FL-algebras. The fact itself follows easily from properties of FL-algebras. The induction step uses Lemma \ref{lem--PRE-PDL validities hold} and is easy; for instance, in the case $\pa = \pb^{+}$ we may use the fact that, for all $x$ and $y$, $V([\pb](\fp \land [\pb^{+}]\fp, x) \sqsubseteq R_{\pb}[x][y] \Rightarrow V([\pb^{+}]\fp, y)$ and hence, for all $s,t$ and $\pi \in \Pi(S)$, $V([\pb^{+}]\fp, s) \sqsubseteq R_{\pb} [s][\pi][t] \Rightarrow V(\fp, t)$ as required.
\end{proof}

\begin{lemma}\label{lem--DEC-filtration}
For all models $\mathfrak{M}$, all $\fp \in \Phi$ and $s \in \mathfrak{M}$, $V(\fp, s) = V^{\Phi}(\fp, [s])$.
\end{lemma}
\begin{proof}
The proof is by induction on the complexity of $\fp$. The base case $\fp \in Pr$ holds by definition, the cases for constants and propositional connectives are trivial and the case $\fp = [\pa]\fq$ is established using Lemma \ref{lem--DEC-filtration lemma}.
\end{proof}

\begin{theorem}\label{thm--decidability}
$Th(\bm{X})$ is decidable for each finite $\bm{X}$.
\end{theorem}
\begin{proof}
Lemma \ref{lem--DEC-filtration} implies $\fp \in Th(\bm{X})$ iff $\fp$ is valid in all frames where $|S| \leq |\bm{X}|^{|\Phi|}$ where $\Phi$ is the closure of $\{ \fp \}$. Now $m: = |\bm{X}| = m$, $n : = m^{|\Phi|}$ and let $n$-frames be the frames with $|S| \leq n$. There are at most 
\[
n \times m^{n^{2}} 
\]
$n$-frames. On each $n$-frame, there are $n \times m^{\omega}$ models, but there are at most $n \times |\Phi| \times m$ possible ways to evaluate elements of $|\Phi|$ on an $n$-frame. Hence, there are at most
\[
m^{n^{2} + 1} \times n^{2} \times |\Phi|
\] 
models to check. It is not hard to show that there is an algorithm checking validity of formulas in finite models.
\end{proof}

\section{Completeness}\label{sec--COM}

Bou et al.\ \cite{Bou2011} establish a general weak completeness result for modal logics based on finite commutative integral FL-algebras with canonical constants  where $0$ is the bottom element. In this section we build on their work to show how a Hilbert-style axiomatic presentation of any finite commutative integral FL-algebra $\bm{X}$ with canonical constants can be extended to a sound and weakly complete axiomatization of PDL based on $\bm{X}$. The restriction to commutative FL-algebras seems to be necessary for our style of argument to go through and we discuss this at appropriate places in more detail; the restriction to integral FL-algebras is convenient. We leave generalizations of our result as an open problem.

Fix a finite commutative integral FL-algebra $\bm{X}$ with canonical constants denoting elements of $\bm{X}$, together with a Hilbert-style axiomatic presentation $\mathsf{Log}(\bm{X})$ in the language $\mathcal{L}_{\bm{X}}$ that is \emph{strongly complete with respect to $\bm{X}$}. That is, we assume that $\fp \in \mathcal{L}_{\bm{X}}$ is derivable from $\Gamma \subseteq \mathcal{L}_{\bm{X}}$ in $\mathsf{Log}(\bm{X})$, in symbols $\Gamma \vdash_{\mathsf{Log}(\bm{X})} \fp$, iff each non-modal homomorphism $u : \mathcal{L}_{\bm{X}} \to \bm{X}$ such that $ 1 \sqsubseteq \bigsqcap u[\Gamma] $ satisfies $ 1 \sqsubseteq u(\fp)$ (values $u([\pa]\fq)$ of modal formulas under $u$ are arbitrary, so $u$ ``treats'' modal formulas as propositional atoms).\footnote{A function $f : \mathcal{L}_{\bm{X}} \to \bm{X}$ is a non-modal homomorphism iff $f(\bar{c}) = c$ and $f$ commutes with the propositional connectives $\oplus$ of $\mathcal{L}_{\bm{X}}$ and the corresponding operations $\oplus^{\bm{X}}$ on $\bm{X}$; we assume that $\land^{\bm{X}}$ is $\sqcap$ and $\lor^{\bm{X}}$ is $\sqcup$.} For the details on how $\mathsf{Log}(\bm{X})$ looks like, see \cite{Bou2011}. Since $\bm{X}$ is finite, $\vdash_{\mathsf{Log}(\bm{X})}$ is \emph{finitary} in the sense that if $\Gamma \vdash_{\mathsf{Log}(\bm{X})} \fp$, then there is a finite $\Delta \subseteq \Gamma$ such that $\Delta \vdash_{\mathsf{Log}(\bm{X})} \fp$. We note that $\vdash_{\mathsf{Log}(\bm{X})}$ is also \emph{monotonic} in the sense that if $\Gamma \vdash_{\mathsf{Log}(\bm{X})} \fp$ and $\Gamma \subseteq \Delta$, then $\Delta \vdash_{\mathsf{Log}(\bm{X})} \fp$.

Since $\bm{X}$ is commutative, we have $a \backslash b = b \slash a$ and so we use only a single ``official'' implication operator $\to$; see \cite[p.\ 95]{Galatos2007}. Recall that $\fp \tot \fq := (\fp \to \fq) \land (\fq \to \fp)$; we define similarly $a \Leftrightarrow b := (a \Rightarrow b) \sqcap (b \Rightarrow a)$.

\begin{definition}
$\mathsf{PDL}(\bm{X})$ is the Hilbert-style axiom system extending $\mathsf{Log}(\bm{X})$ with the following axioms and rules (for all formulas $\fp, \fq$, all action expressions $\pa, \pb \in \mathit{ACT}$ and all canonical constants $\bar{c}$):
\begin{center}
\begin{minipage}[t]{0.4\linewidth}
\begin{tabular}{ll}
(A-$1$) & $[\pa] \bar{1}$\\[1mm]
(A-reg) & $[\pa] \fp \land [\pa]\fq \to [\pa] (\fp \land \fq)$\\[1mm]
(A-$\bar{c}$) 
& $[\pa](\bar{c} \to \fp) \tot (\bar{c} \to [\pa]\fp)$  \\[2mm]
(R-mon) & $\dfrac{\fp \to \fq}{[\pa] \fp \to [\pa] \fq}$

\end{tabular}
\end{minipage}
\qquad
\begin{minipage}[t]{0.4\linewidth}
\begin{tabular}{ll}
(A-$\cup$) & $[\pa \cup \pb]\fp \tot ([\pa]\fp \land [\pb]\fp)$\\[1mm]
(A-$;$) & $[\pa\pb] \fp \tot [\pa][\pb]\fp$\\[1mm]
(A-$+$) & $[\pa^{+}]\fp \tot [\pa](\fp \land [\pa^{+}]\fp)$\\[2mm]
(R-$+$) & $\dfrac{\fp \to [\pa]\fp}{\fp \to [\pa^{+}]\fp}$
\end{tabular}
\end{minipage}
\end{center}
The notions of proof, derivability, theorem and a formula derivable from a set of formulas are defined as usual (see \cite{Bou2011}). $\mathsf{Thm}(\mathsf{PDL}(\bm{X}))$ is the set of theorems of $\mathsf{PDL}(\bm{X})$.
\end{definition}

\noindent Since $\bm{X}$ is fixed, we write $\mathsf{L}$ instead of $\mathsf{Log}(\bm{X})$, $\mathsf{PDL}$ instead of $\mathsf{PDL}(\bm{X})$, $\mathsf{Thm}$ instead of $\mathsf{Thm}(\mathsf{PDL}(\bm{X}))$ and $\mathcal{L}$ instead of $\mathcal{L}_{\bm{X}}$ for the rest of this section.

\begin{theorem}\label{thm--soundness}
If $\fp$ is a theorem of $\mathsf{PDL}$, then $\fp$ is valid in the class of all $\bm{X}$-frames. 
\end{theorem}
\begin{proof}
The axioms and the rule in the left column are taken from \cite{Bou2011}. Validity of the axioms in the right column in all FL-algebras was established in Lemma \ref{lem--PRE-PDL validities hold}. To show that the rule (R-$+$) preserves validity in models, assume that $V(\fp,s) \sqsubseteq V([\pa]\fp, s)$ for all $s$ in an arbitrary model. Take some $t$ and assume that $a \sqsubseteq V(\fp, t)$; we prove that $a \sqsubseteq R_{\pa^{+}} tu \Rightarrow V(\fp, u)$ for all $u$. The claim to be proved is equivalent to $(\forall \pi \in \Pi(S))(a \sqsubseteq R_{\pa} t \pi u \Rightarrow V(\fp, u))$. This claim is easily established by induction on the length of $\pi$.
\end{proof}
\noindent We note that, without the assumption of commutativity, versions of (A-$\bar{c}$) are not sound; the axiom is used in the proof of Lemma \ref{lem--COM-Bou} which is in turn applied in most of our arguments below.

From now on, let $S$ be the set of non-modal homomorphisms $s: \mathcal{L} \to \bm{X}$ such that $s[\mathsf{Thm}] = \{ 1 \}$ and let $\Phi$ be a fixed finite closed set.

\begin{definition}
The \emph{$\Phi$-equivalence relation} on $S$ is an $\bm{X}$-valued binary relation $\sim_{\Phi}$ on $S$ defined by
\[
s \sim_{\Phi} t \quad := \quad \bigsqcap_{\fp \in \Phi} \big ( s(\fp) \Leftrightarrow t(\fp) \big ) \, . 
\]
\end{definition}

\noindent If $\Phi$ is clear from the context, we will write $s \sim t$ or just $st$ instead of $s \sim_{\Phi} t$. 

\begin{lemma}\label{lem--COM-Phi-equivalence}
The relation $\sim_{\Phi}$ is an $\bm{X}$-valued equivalence relation, that is, (a) $1 \sqsubseteq s \sim s$, (b) $s \sim t = t \sim s$ and (c) $(s \sim t)(t \sim u) \sqsubseteq s \sim u$, for all $s,t,u \in S$.
\end{lemma}
\begin{proof}
Claims (a) and (b) are clear; claim (c) follows from Lemma \ref{lem--APP-Properties of L}.
\end{proof}

Completeness proofs for two-valued PDL typically use a filtration-like construction of the canonical model, where states are (or boil down to) equivalence classes of states taken from some other structure. A natural approach in our case would be to take ``equivalence classes'' of non-modal homomorphisms under $\sim$, where $s \sim t$ expresses ``how much equivalent'' $s$ and $t$ are with respect to $\Phi$. However, in our case a simpler approach is available. We take $S$ itself as the set of states of the canonical model and we refer to $\Phi$ only in the definition of the canonical $R_{\pa}$, which is a generalization of the definition of accessibility relations in the greatest filtration of a Kripke model.

\begin{definition}
The \emph{canonical model modulo $\Phi$} is $\mathfrak{M} = (S, R, V)$ where 
\begin{itemize}
\item $S$ is the set of non-modal homomorphisms $s: \mathcal{L} \to \bm{X}$ such that $s[\mathsf{Thm}] = \{ 1 \}$;
\item $R_{m} st := \bigsqcap_{[m]\fp \in \Phi} \big ( s([m]\fp) \Rightarrow t(\fp) \big )$ for all $\apa_m \in Ac$ and $R_{\pa}st$ for $\pa \not\in Ac$ is defined as in models;
\item $V(\afp, s) := s(\afp)$ and $V(\fp, s)$ for $\fp \not\in Pr$ is defined as in models. 
\end{itemize}
We define for each $\pa$ the relation $R^{\mathcal{L}}_{\pa}$ on $S$ by $R^{\mathcal{L}}_{\pa} st := \bigsqcap_{\fp \in \mathcal{L}} \big ( s([\pa]\fp) \Rightarrow t(\fp) \big )$.
\end{definition}

Note that $R^{\mathcal{L}}_{n} st \sqsubseteq R_{n} st$ for all $\apa_n \in Ac$ and all $s,t$ since $R^{\mathcal{L}}_{n}$ ``cares'' about more formulas.   $R^{\mathcal{L}}_{\pa}$ is the usual canonical many-valued accessibility relation, see \cite{Bou2011}, but we cannot use it here because of the presence of the Kleene plus iteration operator in $\mathit{ACT}$, similarly as in the case of two-valued PDL.

The following lemma states some properties of $R^{\mathcal{L}}_{\pa}$ that will be useful in our proofs; the proof of the lemma can be found in \cite{Bou2011} (the logics studied there are mono-modal, but the same approach applies here).

\begin{lemma}\label{lem--COM-Bou}
The following holds for all $\pa \in \mathit{ACT}$ and all $s \in S$ of the canonical model:
\begin{enumerate}[label={(\alph*)}]
\item For all $t$, $R^{\mathcal{L}}_{\pa}st =  \bigsqcap_{\fp \in \mathcal{L}} \big \{ t(\fp) \mid 1 \sqsubseteq  s([\pa]\fp) \big \}$ (\cite{Bou2011}, Proposition 4.1.);
\item For all $\fp \in \mathcal{L}$, $s([\pa]\fp) = \bigsqcap_{u \in S} \big \{ R^{\mathcal{L}}_{\pa} su \Rightarrow u(\fp) \big \}$ (\cite{Bou2011}, Lemma 4.8.).
\end{enumerate}
\end{lemma}

\begin{lemma}\label{lem--COM-Boxed formulas and R}
For all $[\pa]\fp \in \Phi$ and all $s,t \in S$, $s([\pa]\fp) \sqsubseteq R_{\pa} st \Rightarrow t(\fp)$.
\end{lemma}
\begin{proof}
The claim is proved by induction on the complexity of $\pa$. The base case is established as follows. We know that $s([n]\fp) \cdot (s([n]\fp) \Rightarrow t(\fp)) \sqsubseteq t(\fp)$; from this $s([n]\fp) \cdot R_{n} st \sqsubseteq t(\fp)$ follows by the definition of $R_n$.

The cases of choice and composition in the induction step are straightforward. The case $\pa = \pb^{+}$ is established by showing that, for all $\pi \in \Pi(S)$, all $s,t$, and all $\fp$ such that $[\pb^{+}]\fp \in \Phi$, $s([\pb^{+}]\fp) \sqsubseteq R_{\pb} s \pi t \Rightarrow t(\fp)$. This claim, call it (A), follows from the claims ($s,t$ and $[\pb^{+}]\fp \in \Phi$ are fixed)
\begin{itemize}
\item[(B)] $s([\pb]\fp) \sqsubseteq R_{\pb} st \Rightarrow t(\fp)$;
\item[(C)] for all $\sigma \in \Pi(S)$ and all $u$, $s([\pb^{+}]\fp) \sqsubseteq R_{\pb} s \sigma u \Rightarrow u([\pb^{+}]\fp)$.
\end{itemize}
The proof of (C) is left to the reader; (B) holds by the induction hypothesis.
\end{proof}

\begin{lemma}\label{lem--COM--R closed under Phi equivalence}
For all $\pa$ and $s,t,u$, $R_{\pa}  su(ut) \sqsubseteq R_{\pa}  st$.
\end{lemma}
\begin{proof}
We argue by induction on the complexity of $\pa$. The base case is established as follows. If $a \sqsubseteq R_{n}  su(ut)$, then, by definition, $a \sqsubseteq \bigsqcap_{[n]\fp \in \Phi} \big ( s([n]\fp) \Rightarrow u(\fp) \big ) (ut)$. Hence, for all $[n]\fp \in \Phi$, $a \sqsubseteq \big ( s([n]\fp) \Rightarrow u(\fp) \big) \big ( u(\fp) \Rightarrow t(\fp) \big )$ by the definition of $u \sim t$ and monotonicity of monoid multiplication (also, $[n]\fp \in \Phi$ implies $\fp \in \Phi$). It follows by the properties of FL-algebras that $a \sqsubseteq \big ( s([n]\fp) \Rightarrow t(\fp) \big )$. Since $[n]\fp \in \Phi$ was arbitrary, we obtain $a \sqsubseteq R_{n}  st$. All cases of the induction step are easy.
\end{proof}

\begin{definition}
For all $\pa$ and $s$, we define the following formula:
\[
R_{\pa}s \quad := \quad \bigvee_{x  \in S} \Big ( \overline{R_{\pa} s x }\:  \cdot  \bigwedge_{\fp \in \Phi} \big ( \overline{x(\fp)} \tot \fp \big)\Big )
\]
\end{definition}
\noindent Note that $R_{\pa} s$ is well defined even though $S$ is infinite -- there are only finitely many possible values of $R_{\pa} sx$ for $x \in S$, as $\bm{X}$ is finite. Note also that $t( R_{\pa}s) = \bigsqcup_{x \in S} \big ( R_{\pa}  sx(xt) \big)$.

\begin{lemma}\label{lem--COM-RPhi formula}
For all $s,t$ and $\pa$, $t( R_{\pa} s ) = R_{\pa}  st$.
\end{lemma}
\begin{proof}
First, $R_{\pa}  st \sqsubseteq R_{\pa}  st (tt) $ by Lemma \ref{lem--COM-Phi-equivalence}(a), and $R_{\pa}  st(tt) \sqsubseteq \bigsqcup_{x  \in S} \big ( R_{\pa}  sx (xt) \big ) = t(R_\pa s)$. Second, $R_{\pa}  sx (xt) \sqsubseteq R _{\pa} st$ for all $x  \in S$ by Lemma \ref{lem--COM--R closed under Phi equivalence}. Hence, $\bigsqcup_{x \in S} R_{\pa} sx(xt)$ and so $t(R_\pa s) \sqsubseteq R _{\pa} st$.
\end{proof}

\begin{lemma}\label{lem--COM-R under RPhi}
For all $s,t \in S$ and all $\pa \in \mathit{ACT}$, $R^{\mathcal{L}}_{\pa}st \sqsubseteq R_{\pa}  st$.
\end{lemma}
\begin{proof}
Induction on the complexity of $\pa$. The base case follows from definition.  To establish the induction step, we reason by cases. Note that the induction hypothesis is equivalent to the claim that, for all $\pa,\pb$ and $x$, $1 \sqsubseteq x([\pa] R_{\pa} x)$ and $1 \sqsubseteq x([\pb] R_{\pb}x)$ by Lemmas \ref{lem--COM-Bou}(b) and \ref{lem--COM-RPhi formula}.

If $a \sqsubseteq R^{\mathcal{L}}_{\pa \cup \pb} st$, then $a \sqsubseteq \bigsqcap_{\fp \in \mathcal{L}} \big \{ t(\fp) \mid 1 \sqsubseteq s([\pa \cup \pb]\fp) \big\}$ by Lemma \ref{lem--COM-Bou}(a). By the definition of $S$, this entails $a \sqsubseteq \bigsqcap \big\{ t(\fp) \mid 1 \sqsubseteq s([\pa]\fp) \sqcap s([\pb]\fp) \big\}$. By the induction hypothesis, $1 \sqsubseteq s([\pa]R_\pa s)$ and $1 \sqsubseteq s([\pb]R_\pb s)$. Hence, $1 \sqsubseteq s([\pa](R_{\pa}s \lor R_{\pb}s))$ and $1 \sqsubseteq s([\pb](R_{\pa}s \lor R_{\pb}s))$ by the definition of $S$. It follows that $a \sqsubseteq t(R_{\pa}s) \sqcup t(R_{\pb}s)$. By Lemma \ref{lem--COM-RPhi formula}, $a \sqsubseteq R_{\pa}st \sqcup R_{\pb}st$ and so $a \sqsubseteq R_{\pa \cup \pb} st$.

If $a \sqsubseteq R^{\mathcal{L}}_{\pa\pb} st$, then $a \sqsubseteq \bigsqcap_{\fp \in \mathcal{L}} \big \{ t(\fp) \mid 1 \sqsubseteq s([\pa\pb] \fp) \big\}$ by Lemma \ref{lem--COM-Bou}(a) and so $a \sqsubseteq \bigsqcap_{\fp \in \mathcal{L}} \big \{ t(\fp) \mid 1 \sqsubseteq s([\pa][\pb] \fp) \big\}$ by the definition of $S$. For all $x$ and $y$, $R_{\pa}^{\mathcal{L}}sx R_{\pb}^{\mathcal{L}} xy \sqsubseteq y(R_{\pa\pb}s)$ by the induction hypothesis, Lemma \ref{lem--COM-RPhi formula} and the definition of $R_{\pa\pb}$. Hence, for all $x$, $R_{\pa}^{\mathcal{L}}sx \sqsubseteq x ([\pb]R_{\pa\pb}s)$ by residuation and Lemma \ref{lem--COM-Bou}(b); from this is follows that $1 \sqsubseteq s ([\pa][\pb]R_{\pa\pb}s)$ by another application of residuation and Lemma \ref{lem--COM-Bou}(b). Therefore, $a \sqsubseteq t(R_{\pa\pb} s)$ and so $a \sqsubseteq R_{\pa\pb} st$ by Lemma \ref{lem--COM-RPhi formula}. 

Finally, we discuss the case of $\pa^{+}$. Fix $s$; we write $F$ instead of $R_{\pa^{+}} s$. Note that $R _{\pa^{+}}$ is a transitive relation extending $R^{\mathcal{L}}_{\pa}$. Hence, for all $t,u \in S$, $u(F) \cdot R^{\mathcal{L}}_{\pa} ut \sqsubseteq t(F)$ by Lemma \ref{lem--COM-RPhi formula} and the induction hypothesis applied to $R^{\mathcal{L}}_{\pa}$; we obtain from this $u(F) \sqsubseteq u([\pa]F)$ for all $u \in S$ by Lemma \ref{lem--COM-Bou}(b). Hence, by definition of $S$, we have $ F \to [\pa]F \: \in \: \mathsf{Thm} $. Hence, using (R-$+$), we have $F \to [\pa^{+}]F \: \in \: \mathsf{Thm}$ and, using (R-mon) and (A-$+$), we obtain $[\pa]F \to [\pa^{+}]F \in \mathsf{Thm}$. By the induction hypothesis we have $R_{\pa}^{\mathcal{L}} st \sqsubseteq R_{\pa} st \sqsubseteq R_{\pa^{+}} st$ for all $t$ and so $1 \sqsubseteq R_{\pa}^{\mathcal{L}} st \Rightarrow t(F)$ for all $t$ by Lemma \ref{lem--COM-RPhi formula}. This means that $1 \sqsubseteq s([\pa]F)$ and so $1 \sqsubseteq s([\pa^{+}]F)$ which means that $R^{\mathcal{L}}_{\pa^{+}} st \sqsubseteq t(F)$ for all $t$ by Lemma \ref{lem--COM-Bou}(b). Hence, $R^{\mathcal{L}}_{\pa^{+}} st \sqsubseteq R _{\pa^{+}} st$ by Lemma \ref{lem--COM-RPhi formula}.
\end{proof}

\begin{lemma}\label{lem--COM-Truth lemma}
For all $\fp \in \Phi$ and $s \in S$, $s(\fp) = V(\fp, s )$.
\end{lemma}
\begin{proof}
Induction on the complexity of $\fp$. The base case holds by definition and the cases for non-modal formulas and canonical constants are straightforward. Finally, $s([\pa]\fp) \sqsubseteq V([\pa]\fp, s )$ holds thanks to Lemma \ref{lem--COM-Boxed formulas and R} and $V([\pa]\fp, s ) \sqsubseteq s([\pa]\fp)$ holds thanks to Lemma \ref{lem--COM-Bou}(b) and Lemma \ref{lem--COM-R under RPhi}.
\end{proof}

\begin{theorem}\label{thm--completeness}
For all finite commutative integral $\bm{X}$ with canonical constants, $\fp$ is valid in all $\bm{X}$-frames iff $\fp$ is a theorem of $\mathsf{PDL}(\bm{X})$.
\end{theorem}
\begin{proof}
Soundness is established by Theorem \ref{thm--soundness}. Completeness is established as usual. If $\fp$ is not in $\mathsf{Thm}$, then $\mathsf{Thm} \not\vdash_{\mathsf{L}} \fp$ since $\mathsf{Thm}$ is obviously closed under $\vdash_{\mathsf{L}}$. By strong completeness of $\mathsf{L}$, there is a non-modal homomorphism from $\mathcal{L}$ to $\bm{X}$ such that $s[\mathsf{Thm}] = \{ 1 \}$ and $s(\fp) \neq 1$. Let $\Phi$ be the closure of $\{ \fp \}$; $\fp$ is not valid in the canonical model modulo $\Phi$ by Lemma \ref{lem--COM-Truth lemma}.
\end{proof}

\section{On Kleene star and test}\label{sec--ISS}

Our syntactic presentation of propositional dynamic logic differs from the standard presentation in two important respects, namely, (i) our action operators do not include the \emph{Kleene star}, but rather the Kleene plus operator; (ii) we do not include the \emph{test operator}. Kleene star and test are instrumental in the ability of classical PDL to express standard programming constructs such as while loops and conditionals (test suffices for the latter). In this section we discuss these omissions.

Concerning the Kleene star, Proposition \ref{prop--PRE-Reflexive transitive closure} suggests that, working with frames based on finite integral FL-algebras, we can define, for all $\pa \in \mathit{ACT}$ and $\fp \in Fm(\mathcal{L}_{\bm{X}})$, \[ [\pa^{*}]\fp := [\pa^{+}]\fp \land \fp \] as a semantically equivalent surrogate for formulas with the Kleene star. For instance, $[(\apa \cup \apb)^{*} ; \apa^{*}]\afp$ is short for $[(\apa \cup \apb)^{+}]([\apa^{+}]\afp \land \afp) \land ([\apa^{+}]\afp \land \afp)$. However, it is clear that not all action expressions in $\mathit{STA}$ can be expressed by action expressions in $\mathit{ACT}$. Therefore, for example, $[(\apa^{*}; \apb)^{*}]\afp$ is not a well-formed formula since $\apa^{*} \not\in \mathit{ACT}$. 

The technical problem that precluded us from working with Kleene star as a primitive operator is related to Lemma \ref{lem--COM--R closed under Phi equivalence}. Take the reflexive transitive closure $R^{*}_{\pa}$ of $R_{\pa}$, defined as in Proposition \ref{prop--PRE-closure}. The issue is that Lemma \ref{lem--COM--R closed under Phi equivalence} fails if Kleene star is a primitive operator and we define $R_{\pa^{*}} := R_{\pa}^{*}$. In particular, if $s = u \neq t$, then $R^{*}_{\pa}su(ut) \sqsubseteq R_{\pa}^{*}st$ boils down to $s \sim t \sqsubseteq R_{\pa^{+}}st$, which does not hold in all canonical models. (Take the canonical $\bm{2}$-model modulo the closure $\Phi$ of $\Psi = \{ [\apa]\bot \}$. As both $\Psi \cup \{ \afp_0 \}$ and $\Psi \cup \{ \afp_1 \}$ are consistent, there are two distinct $s,t$ such that $s \sim_{\Phi} t$ equals $1$, but $R_{\apa^{+}}st$ equals $0$.)

Concerning test, a natural semantic interpretation of $\fp ?$, endorsed also in \cite{Hughes2006,Liau1999}, is 
\[
R_{\fp ?} (s,t) = \begin{cases}
V(\fp, s) & \text{if } s = t \\
\bot^{\bm{X}} & \text{otherwise.}
\end{cases}
\]
However, Lemma \ref{lem--COM--R closed under Phi equivalence} turns out to be problematic for such a relation as well. (Take the model from the previous paragraph and let $\fp = [\apa]\bot$; clearly $R_{\fp ?} ss (st)$ equals $1$, but $R_{\fp ?} st$ equals $0$.)

It is clear that a more substantial modification of our completeness argument is needed to accommodate logics with Kleene star and test. This is an interesting problem we leave open here.

\section{Conclusion}\label{sec--CON}

We have studied a general framework for many-valued versions of Propositional Dynamic Logic where both formulas in states and accessibility relations between states of a Kripke model are evaluated in a finite FL-algebra. We established a general decidability result and we provided a general completeness argument for PDLs based on commutative integral FL-algebras with canonical constants. We build on previous work on many-valued modal logic and our techniques are generalizations of the arguments used in the two-valued case; however, to the best of our knowledge, the technical results presented here are the first decidability and completeness results on PDL with many-valued accessibility relations. As our discussion of the informal interpretations of the framework suggests, many-valued PDL has links to existing research in description logics and potential applications in reasoning about weighted labelled transition systems.

Our paper also suggests a number of topics for future research. We would like to mention especially the addition of test and further work on the standard version of PDL with primitive Kleene star in the many-valued setting. Another topic are generalizations of our results beyond finite (commutative integral) FL-algebras with canonical constants; in many cases the work here would require modifications of  existing techniques used in completeness arguments for many-valued modal logics without ``structured'' modal operators. Finally, informal interpretations and applications of our framework need to be explored in more detail.



\end{document}